\documentclass[conference, 12pt, onecolumn]{ieeeconf}

\IEEEoverridecommandlockouts
\usepackage[usenames]{color}
\usepackage{enumerate}
\usepackage{url}
\usepackage{subfigure}
\usepackage{amsfonts,mathrsfs}
\usepackage{amssymb,amsmath}
\usepackage{verbatim}
\usepackage{acronym}
\usepackage{mathtools}
\usepackage{cite}
\usepackage{graphicx}
\usepackage{algorithm}
\usepackage[noend]{algpseudocode}

\def\fskip#1{}

\newtheorem{theorem}{Theorem}

\newtheorem{conjecture}[theorem]{Conjecture}

\newtheorem{definition}{Definition}

\newtheorem{lemma}{Lemma}

\newtheorem{proposition}[theorem]{Proposition}
\newtheorem{remark}{Remark}

\renewcommand{\theenumi}{\arabic{enumi}}

\def\1{{\bf 1}}

\def\d{\delta}

\newcommand{\remove}[1]{}

\def\argmax{\mathop{\rm argmax}}

\begin{document}
\title{A Fixed Point Framework for the Existence of EFX Allocations}
%\vspace{-0.5cm}
\author{\authorblockN{S. Rasoul Etesami*}\thanks{*Department of Industrial and Systems Engineering, Department of Electrical and Computer Engineering, and Coordinated Science Laboratory, University of Illinois Urbana-Champaign,  Urbana, IL 61801 (Email: etesami1@illinois.edu).}
%\thanks{This research was supported by the AFOSR YIP FA9550-23-1-0107.}
}
\maketitle

\begin{abstract}
We consider the problem of the existence of an envy-free allocation up to any good (EFX) for linear valuations and establish new results by connecting this problem to a fixed-point framework. Specifically, we first use randomized rounding to extend the discrete EFX constraints into a continuous space and show that an EFX allocation exists if and only if the optimal value of the continuously extended objective function is nonpositive. In particular, we demonstrate that this optimization problem can be formulated as an unconstrained difference-of-convex (DC) program, which can be further simplified to the minimization of a piecewise linear concave function over a polytope. Leveraging this connection, we show that the proposed DC program has a nonpositive optimal objective value if and only if a well-defined continuous vector map admits a fixed point. Crucially, we prove that the reformulated fixed-point problem satisfies all the conditions of Brouwer’s fixed-point theorem, except that self-containedness is violated by an arbitrarily small positive constant. To address this, we propose a slightly perturbed continuous map that always admits a fixed point. This fixed point serves as a proxy for the fixed point (if it exists) of the original map, and hence for an EFX allocation through an appropriate transformation. Our results offer a new approach to establishing the existence of EFX allocations through fixed-point theorems. Moreover, the equivalence with DC programming enables a more efficient and systematic method for computing such allocations (if one exists) using tools from nonlinear optimization. Our findings bridge the discrete problem of finding an EFX allocation with two continuous frameworks: solving an unconstrained DC program and identifying a fixed point of a continuous vector map.
\end{abstract}

% REQUIRED

\section{Introduction}

The fair division of indivisible goods has long been a central topic in economics, game theory, and algorithmic mechanism design \cite{caragiannis2019envy,akrami2025efx,amanatidis2023fair}. Unlike divisible resources (e.g., in cake-cutting), the indivisibility of goods introduces unique challenges, as standard fairness notions such as \emph{envy-freeness (EF)}—where no agent prefers another's bundle of items to their own—may not always be achievable \cite{plaut2018almost}. In such cases, relaxed notions of fairness have been proposed to provide more practical guarantees.

One such relaxation, \textit{envy-freeness up to one item (EF1)}, introduced by~\cite{budish2011combinatorial}, guarantees that no agent envies another after the hypothetical removal of a single good from the other's bundle. EF1 has been widely accepted as a fairness requirement, and its existence, along with efficient algorithms for computing EF1 allocations, has made it a practical benchmark for fairness~\cite{lipton2004approximately}. However, EF1 may still permit significant envy in certain cases, motivating the study of a stronger fairness criterion known as \textit{envy-freeness up to any item (EFX)}. Formally introduced by~\cite{caragiannis2019unreasonable}, an allocation is EFX if no agent envies another after the removal of \textit{any} single good from the other's bundle. Thus, EFX is a strong relaxation of envy-freeness in the context of indivisible goods and serves as a close approximation to EF allocations. Despite its strong theoretical appeal, however, \textit{the existence of EFX allocations remains one of the most important open problems in the fair division of indivisible goods}. Unlike EF1, EFX allocations are not known to exist in all cases, and the general question of their existence for arbitrary additive valuations with more than three agents remains unresolved.

\subsection{Related Work}

The work of~\cite{plaut2018almost} used the leximin solution to establish the existence of EFX allocations in several contexts, sometimes in conjunction with Pareto optimality—a well-known notion of economic efficiency. \cite{caragiannis2019unreasonable} provided a key structural insight by showing that \textit{Nash social welfare} (NSW) optimal allocations—those that maximize the geometric mean of utilities—are EF1 and Pareto optimal when agents have additive valuations (see also~\cite{barman2020finding}). This connection between fairness and efficiency has renewed interest in exploring the existence of allocations that are approximately EFX-fair and NSW-optimal~\cite{feldman2024optimal}.

Since the existence of EFX allocations is not known in general, subsequent work has sought to identify settings where EFX allocations are guaranteed to exist. For instance,~\cite{plaut2018almost} established the existence of EFX allocations for two agents with arbitrary monotone valuations, as well as for many agents with identical general monotone valuations. However, they also argued that extending the result to three or more agents poses significant challenges. \cite{babaioff2021fair} and~\cite{bu2023efx} showed the existence of EFX allocations for binary valuations, where the marginal gain in value from receiving an additional item is either $0$ or $1$. \cite{chaudhury2024efx} established the existence of EFX allocations for three agents with additive valuations using constructive methods. This result was extended in~\cite{akrami2025efx}, which showed the existence of EFX allocations for three agents with valuations more general than additive—further supporting the conjecture that EFX allocations may exist in broader settings. 

However, generalizing these results to more than three agents and to non-identical valuations has proven elusive. The work of~\cite{benade2023existence} considered stochastic valuations beyond the additive setting and derived conditions on the number of agents and goods under which EFX allocations exist with high probability. Moreover,~\cite{christodoulou2023fair} established the existence of EFX allocations in settings where valuations are represented via a graph of arbitrary size, where vertices correspond to agents and edges to items. In their setting, an item (edge) has zero marginal value to all agents (vertices) not incident to it, and each vertex may have an arbitrary monotone valuation over the set of incident edges.

From a computational perspective, it is known that computing EF allocations, even when part of the allocation is fixed and the task is to assign the remaining items, is an NP-hard problem~\cite{deligkas2025complexity}. Moreover, the problem of finding EFX allocations has been shown to be nontrivial, even in cases where existence is guaranteed. Currently, there is no known polynomial-time algorithm for computing EFX allocations for three or more agents with additive valuations~\cite{amanatidis2023fair}. In addition,~\cite{plaut2018almost} showed that, for general valuation functions, any deterministic algorithm requires an exponential number of value queries to find an EFX allocation.

To bypass known hardness results and the potential non-existence of EFX allocations in general, recent work has considered allowing some items to remain unallocated—a concept known as \textit{EFX with charity} or \textit{partial EFX}. That is, by donating some items to a charity, one can distribute the remaining items in a fair way. This relaxation ensures fairness among the participating agents, even if it means not allocating all goods. \cite{caragiannis2019envy} first showed that under additive valuations, there exists an EFX allocation of a subset of items with a NSW that is at least half of the maximum possible NSW for the original set of items. Subsequently, a line of work has sought to reduce the number of unallocated items as much as possible. \cite{chaudhury2021little} showed that for general valuations and $n$ agents, there always exists an EFX allocation that sends fewer than $n$ items to charity, and no agent values the charity items more than her own bundle. This line of work highlights a compelling trade-off between fairness and completeness, suggesting that a limited sacrifice in allocative efficiency can yield stronger fairness guarantees.

Given the challenges of computing or even verifying EFX allocations for general instances, researchers have proposed several approximation frameworks. Instead of focusing solely on exact EFX allocations, a growing line of work has explored the computation of allocations that are approximately EFX, under various notions of approximation. One prominent example is the notion of $\alpha$-EFX allocation, which provides a multiplicative approximation in terms of the values obtained by the agents. Specifically, an allocation is called $\alpha$-EFX if, for every agent, the value they assign to their own bundle is at least $\alpha$ times the value they assign to any other agent’s bundle after the removal of any single good from that bundle. The work \cite{plaut2018almost} was the first to study this notion, showing that $\frac{1}{2}$-EFX allocations always exist, even for subadditive valuation functions. Later,~\cite{chan2019maximin} showed that such allocations can be computed in polynomial time. In fact, for agents with linear valuation functions, the so-called \emph{Envy-Cycle Elimination algorithm} can be used to compute a $\frac{1}{2}$-EFX allocation. Along similar lines, approximate EFX allocations with a small number of unallocated goods (charity) were considered in~\cite{akrami2025efx}, where the existence of a $(1 - \epsilon)$-EFX allocation with at most $\tilde{O}\big((n/\epsilon)^{1/2}\big)$ charity items was established.

In sharp contrast to prior work, which often uses algorithmic constructions to prove the existence of an EFX allocation or to approximate it under certain assumptions, this work adopts a probabilistic method to study the existence of an EFX allocation. Our approach reduces the problem of finding an EFX allocation with linear valuations to solving a continuous nonlinear program. Aside from a parametric integer linear programming formulation for the existence of EFX allocations~\cite{bredereck2023high}, we are not aware of any prior work that uses probabilistic methods and continuous extensions for analyzing the existence of EFX allocations. We refer to \cite{amanatidis2023fair} for a comprehensive survey of recent results on the fair division of indivisible goods.

\vspace{-0.2cm}
\subsection{Contributions and Organization}

In this work, we take a fundamentally different approach from the existing literature by formulating the existence of an EFX allocation as a fixed point problem. To that end, using a probabilistic argument, we first provide a continuous extension of the EFX constraints based on row-wise randomized rounding. Leveraging this continuous extension, we establish an equivalence between the existence of an EFX allocation and the sign of the optimal objective value of a constrained nonlinear program. We then apply an appropriate change of variables to show that this constrained nonlinear program can be reformulated as an unconstrained difference-of-convex (DC) program. This reformulation enables the use of computational tools from DC programming to systematically search for EFX allocations. Further, we provide an alternative characterization of the existence of an EFX allocation in terms of a fixed point of a continuous vector map. This map satisfies all the conditions of the Brouwer’s fixed-point theorem, except that it violates the self-containment property by an arbitrarily small positive constant. To address this, we propose a slightly perturbed continuous map that always admits a fixed point. This fixed point can serve as a proxy for the fixed point of the original map (if one exists), and therefore, for an EFX allocation. As a result, our work bridges the discrete problem of finding an EFX allocation with the continuous problem of solving an unconstrained DC program, as well as with the problem of finding a fixed point of a continuous vector map.

The rest of the paper is organized as follows. Section~\ref{sec:problem} introduces the problem setting. A preliminary convex program relaxation for EFX allocations, based on the Lovász extension, is presented in Section~\ref{sec:lovasz}. Building on the insights obtained, Section~\ref{sec:row} provides a continuous extension of the EFX allocations using row-wise independent rounding. In Section~\ref{sec:DC}, we formulate the EFX problem as a DC program. A fixed-point formulation for EFX allocations is given in Section~\ref{sec:fixed}. Finally, conclusions and open directions for future research are discussed in Section~\ref{sec:con}.

\section{Problem Formulation}\label{sec:problem}

Consider a set of indivisible goods (items) $M = \{1, 2, \ldots, m\}$ and a set of agents $N = \{1, 2, \ldots, n\}$. Each agent $i$ has a nonnegative valuation for any bundle $S \subseteq M$ of items, denoted by $v_i(S)$. An envy-free allocation up to any good (EFX) is a partition $X = (X_1, \ldots, X_n)$ of the goods $M$ among the $n$ agents such that no agent envies the bundle of any other agent after the removal of \emph{any} single good from that bundle. More precisely, the partition $X = (X_1, \ldots, X_n)$ is an EFX allocation if
\begin{align}\label{eq:EFX-original}
v_i(X_i)\ge v_i(X_j\setminus \{k\}),\ \  \forall k\in X_j,\  \forall i,j\in N.
\end{align} 

\begin{definition}
The valuation functions are called \emph{linear} if, for any $i \in N$ and $S \subseteq M$, we have $v_i(S) = \sum_{k \in S} v_{ki}$, where $v_{ki} = v_i(\{k\}) \ge 0$ is the value that agent~$i$ assigns to item~$k$. 
\end{definition} 

The existence of EFX allocations is not known even for linear valuations with more than three agents, which are arguably the most commonly used type of valuations~\cite{amanatidis2023fair}. In this work, we focus on the case of \emph{linear valuations}, under which the EFX conditions~\eqref{eq:EFX-original} can be expressed in a simpler form:
\begin{align}\nonumber
v_i(X_i)\ge v_i(X_j)-v_{ki},\ \  \forall k\in X_j, \forall i,j\in N,
\end{align}
or equivalently as
\begin{align}\label{eq:condition-linear-efx}
v_i(X_i)\ge v_i(X_j)-\min_{k\in X_j} v_{ki},\  \forall i,j\in N \iff  \sum_{k\in X_i}v_{ki}\ge \sum_{k\in X_j}v_{ki}-\min_{k\in X_j} v_{ki},\  \forall i,j\in N.
\end{align}
Therefore, given the agents' valuations of items $\{v_{ki} \ge 0\}$, our goal is to determine whether there exists an allocation $X = (X_1, \ldots, X_n)$ of items to the agents that satisfies condition~\eqref{eq:condition-linear-efx}.

\section{A Convex Program Relaxation Using Lov\'asz Extension}\label{sec:lovasz}

In this section, we provide a preliminary convex program based on Lov\'asz extension of submodular functions to generate a fractional solution to the EFX problem. Inspired from this idea, in the subsequent section, we show how to use a different continuous extension to reduce the EFX problem to a fixed point problem. Moreover, the notations and definitions introduced in this section will be used throughout the paper.

Given an agent $i$ and its valuations $\{v_{ki} \ge 0, k\in M\}$, define the set functions $f_i, v_i:2^M\to \mathbb{R}_+$ as
\begin{align}\label{eq:fi-def}
&f_i(S)=\sum_{k\in S}v_{ki}-\min_{k\in S}v_{ki} \ \ \forall S\subseteq M,\cr
&v_i(S)=\sum_{k\in S}v_{ki} \ \forall S\subseteq M.
\end{align}
Note that $f_i$ is a submodular function, because for any $S \subset T \subset M$ and any item $\ell \notin T$, we have
\begin{align}\nonumber
f_i(S\cup\{\ell\})-f_i(S)&=v_{\ell i}-\min_{k\in S\cup\{\ell\}}v_{ki}+\min_{k\in S}v_{ki}\cr
&=v_{\ell i}+(\min_{k\in S}v_{k i}-v_{\ell i})^+\cr
&\ge v_{\ell i}+(\min_{k\in T}v_{ki}-v_{\ell i})^+\cr
&= f_i(T\cup\{\ell\})-f_i(T),
\end{align}
where $(a)^+\!=\max\{a,0\}$ for a real number $a$. For an arbitrary subset $X\!=(X_1,\ldots,X_n)$ of $M^n$, and any two agents $i\neq j$, define the set function $u_{ij}:2^{M^n}\to \mathbb{R}$ as\footnote{Here, we do not impose the constraint that $X$ must be a partition of $M$, and each $X_i$ can be chosen freely as a subset of $M$.
}
\begin{align}\label{eq:u_ij}
u_{ij}(X)=f_i(X_j)-v_i(X_i),
\end{align}
and note that $u_{ij}$ is a submodular set function over $M^n$ because $f_i$ is submodular and $v_i$ is a modular function. This leads us to the following lemma, which provides an equivalent characterization of EFX allocations and follows directly from the above derivations.

\begin{lemma}\label{lemm:u}
An allocation profile $X = (X_1, \ldots, X_n) \subset M^n$ is an EFX allocation if and only if $X$ is a partition of $M$ and $u_{ij}(X) \leq 0$ for all $i\neq j$, where $u_{ij}$ are the submodular functions defined in~\eqref{eq:u_ij}.
\end{lemma} 

Unfortunately, the submodular functions in~\eqref{eq:u_ij} are neither monotone nor nonnegative, the properties that are often useful in the context of submodular optimization. However, this can be easily fixed once we impose the partition constraints. More specifically, according to Lemma \ref{lemm:u}, we are interested in feasible allocations that \emph{partition} the set of goods $M$. Let us define $V_i=\sum_{k=1}^m v_{ki}\ \forall i\in N$, and note that for any feasible allocation $X$, we have
\begin{align}\label{eq:V_i}
\sum_{j=1}^n v_i(X_j)=\sum_{j=1}^n\sum_{k\in X_j} v_{ki}=\sum_{k=1}^m v_{ki}=V_i. 
\end{align}
Without loss of generality, let us assume that $v_{ki}$ are normalized as $\frac{v_{ki}}{V_i}$ so that $V_i=1\ \forall i$. Therefore, from \eqref{eq:V_i} we have $v_i(X_i)=1-\sum_{\ell\neq i}v_i(X_{\ell})$. Substituting this relation into $u_{ij}$ and using Lemma \ref{lemm:u}, we obtain the following proposition.
\begin{proposition}\label{prop:normalized}
An allocation $X=(X_1,\ldots,X_n)$ is EFX if and only if 
\begin{align}\label{eq:equivalent-discrete}
\min\big\{F(X)=\max_{i\neq j} u_{ij}(X): X \ \mbox{is a partition of}\ M\big\}\leq 1,
\end{align} 
where $u_{ij}(X)=f_i(X_j)+\sum_{\ell\neq i}v_i(X_{\ell})$ are nonnegative monotone submodular functions in which $f_i$ and $v_i$ are defined by \eqref{eq:fi-def} with $v_{ki}$  replaced by their normalized values $\frac{v_{ki}}{V_i}$.      
\end{proposition}

Since solving the discrete optimization problem \eqref{eq:equivalent-discrete} is difficult in general, we instead consider a continuous relaxation of that problem. To that end, we consider the following continuous extension of a submodular function, which will be used to provide a ``tight" convex relaxation of the problem \eqref{eq:equivalent-discrete}.

\begin{definition}\label{def:lovasz}
The Lov\'asz extension of a submodular function $f:2^d\to \mathbb{R}$, denoted by $f^L:\mathbb{R}^d\to \mathbb{R}$, is defined as follows. Let $x\in \mathbb{R}^d$ be an arbitrary vector, and let $\pi:[d]\to [d]$ be the sorted permutation of the coordinates of $x$, i.e., $x_{\pi_1}\ge x_{\pi_2}\ge \ldots\ge x_{\pi_d}$. Let $S^1,\ldots,S^d$ be the prefix sets of this permutation, i.e., $S^i=\{\pi_1,\pi_2,\ldots,\pi_i\}, i=1,\ldots,d$. The value of the Lov\'asz extension at $x$ is given by     
\begin{align}\nonumber
f^L(x)=\sum_{i=1}^d\Big(f(S^i)-f(S^{i-1})\Big)x_{\pi_i},
\end{align} 
where by convention we let $f(S^0)=0$.  
\end{definition}

To provide a convex relaxation for the EFX problem, let $x\in [0,1]^{m\times n}$ be an $m\times n$ matrix whose entry $x_{ki}$ represents the fraction of item $k$ that is allocated to agent $i$. Since we want $x$ to be a fractional relaxation of a feasible allocation $X$, we impose the constraints that $x$ must belong to the feasible polytope $P:=\{x\in [0,1]^{m\times n}: \sum_{i=1}^n x_{ki}=1, \forall k\}$. Moreover, since $u_{ij}(X)$ is a submodular function, we can consider its lov\'asz extension, denoted by $u^L_{ij}(x)$, which is a piecewise linear convex function, and is known to be the ``tightest" convex extension of $u_{ij}$ whose values coincide at any integral point \cite{lovasz1983submodular}. That is $u^L_{ij}(\boldsymbol{1}_{X})=u_{ij}(X)\ \forall X\subseteq M^n$, where $\boldsymbol{1}_{X}$ is the indicator function of the set $X$, and for any convex function $g$ such that $g(\boldsymbol{1}_{X})=u_{ij}(X)\ \forall X$, we must have $g(x)\leq u^L_{ij}(x)$. Now let us define 
\begin{align}\nonumber
f(x)=\max_{i\neq j}u^L_{ij}(x),
\end{align}
and note that $f$ is also a piecewise linear convex function such that $f(\boldsymbol{1}_{X})=F(X)$ for any $X\subseteq M^n$. Therefore, a convex program relaxation for the optimization problem \eqref{eq:equivalent-discrete} is given by
\begin{align}\label{eq:opt-convex-relax}
f^*=\min\{f(x): x\in P\},  
\end{align}
whose optimal solution $x^*$ can be obtained efficiently using a standard subgradient method. Clearly, if $f^*>1$, by Proposition \ref{prop:normalized} we conclude that no EFX allocation exists because \eqref{eq:opt-convex-relax} is a relaxation of the problem \eqref{eq:equivalent-discrete}. Fortunately, it is not difficult to show that $f^*\leq 1 $, as shown in the following. 

\begin{proposition}
Given $x\in P$ and two arbitrary columns $i\neq j\in [n]$, let us sort the items' indices such that $x_{1j}\ge \cdots\ge x_{mj}$. The Lov\'asz extension of $u_{ij}$ is given by the  piecewise linear convex function:
\begin{align}\nonumber
u^L_{ij}(x)=\sum_{k=2}^m\Big(v_{ki}+\big(\min_{r\in [k-1]}v_{ri}-v_{ki}\big)^+\Big)x_{kj}-\sum_{k=1}^mv_{ki}x_{ki}, 
\end{align}
where $[k]=\{1,2,\ldots,k\}$ for an integer $k$. In particular, $f^*\leq 1-\frac{1}{n}\min_{i,r}v_{ri}$. 
\end{proposition}
\begin{proof}
Given the sorted vector of the $j$th column $x_{1j}\ge x_{2j}\ge\cdots\ge x_{mj}$, first we note that 
\begin{align}\nonumber
v_i^L(x_j)=\sum_{k=1}^m\big(v_i([k])-v_i([k-1])\big)x_{kj}=\sum_{k=1}^mv_{ki}x_{kj}.
\end{align}
Similarly, we have $v_i^L(x_i)=\sum_{k=1}^mv_{ki}x_{ki}.$ In order to find the Lov\'asz extension of $-\min_{k\in X_j}v_{ki}$, we have
\begin{align}\label{eq:min-telescope-to_plus}
\sum_{k=1}^m \Big(-\min_{r\in [k]}v_{ri}+\min_{r\in [k-1]}v_{ri}\Big)x_{kj}=-v_{1i}x_{1j}+\sum_{k=2}^m \big(\min_{r\in [k-1]}v_{ri}-v_{ki}\big)^+x_{kj}.
\end{align}
Summing all the above relations and noting that $u^L_{ij}(x)=f^L_i(x_j)-v^L_i(x)$ for any $x\in P$, we obtain the desired result. 

To show $f^*\leq 1$, let $\hat{x}$ be the fractional solution with all entries equal $\frac{1}{n}$, and note that $\hat{x}\in P$. Then
\begin{align}\nonumber
u^L_{ij}(\hat{x})&=\frac{1}{n}\sum_{k=2}^m\Big(v_{ki}+\big(\min_{r\in [k-1]}v_{ri}-v_{ki}\big)^+\Big)-\frac{1}{n}\sum_{k=1}^mv_{ki}+1\cr
&=-\frac{1}{n}v_{1i}+\frac{1}{n}\sum_{k=2}^m\big(\min_{r\in [k-1]}v_{ri}-v_{ki}\big)^++1\cr
&=\frac{1}{n}\sum_{k=1}^m \big(-\min_{r\in [k]}v_{ri}+\min_{r\in [k-1]}v_{ri}\big)+1\cr
&=1-\frac{1}{n}\min_{r\in [m]}v_{ri}<1,
\end{align} 
where the second equality holds by \eqref{eq:min-telescope-to_plus}, and the last equality is obtained using a telescopic sum. Thus, 
\begin{align}\nonumber
f^*\leq f(\hat{x})=\max_{i\neq j}u^L_{ij}(\hat{x})=\max_{i\neq j}\big\{1-\frac{1}{n}\min_{r\in [m]}v_{ri}\big\}=1-\frac{1}{n}\min_{i,r}v_{ri}\leq 1.
\end{align}  
\end{proof} 

Now suppose we solve the convex program \eqref{eq:opt-convex-relax} to obtain an optimal fractional solution $x^* \in [0,1]^{mn}$. One way to round this solution to an integral allocation $X^* = (X_1^*, \ldots, X_n^*)$ is as follows: for each column $i$, independently pick a uniformly random variable $\theta_i \sim \mathrm{Uniform}[0,1]$, and include item $k$ in $X_i^*$ if and only if $x_{ki} \ge \theta_i$. Then, by the definition of the Lov\'asz extension, we have $\mathbb{E}[u_{ij}(X^*)] = u_{ij}^{L}(x^*) \leq 1$. Unfortunately, the quantity we aim to upper bound is $F(X^*) = \max_{i \neq j} u_{ij}(X^*)$, whose expected value could be larger than $\max_{i \neq j} \mathbb{E}[u_{ij}(X^*)] \leq 1$. One might attempt to argue that if the random variables $u_{ij}(X^*)$, for $i \neq j$, are highly concentrated around their means, then we can approximate
\[
\mathbb{E}\big[\max_{i \neq j} u_{ij}(X^*)\big] \approx \max_{i \neq j} \mathbb{E}[u_{ij}(X^*)] \leq 1,
\]
which would imply the existence of an integral allocation $\hat{X}$ such that $\max_{i \neq j} u_{ij}(\hat{X}) \leq 1$. However, the above rounding procedure does not introduce sufficient independence to yield such strong concentration bounds. Even if this issue could be resolved, there remains the problem of feasibility: the resulting $X^*$ may not form a valid partition of $M$. Specifically, under this rounding scheme, an item might be allocated to multiple agents or to none at all. Thus, a form of contention resolution is required to transform the rounded solution $X^*$ into a feasible allocation $\hat{X}$, without significantly increasing the maximum expected utility. Therefore, while an optimal solution to the Lov\'asz extension relaxation may yield a good fractional solutions, it is unclear how to design a rounding scheme that offers provable guarantees. This motivates us to consider an alternative continuous extension of the EFX conditions \eqref{eq:condition-linear-efx} that avoids the limitations of the Lov\'asz extension approach discussed in this section.
        
\section{Independent Row-wise Continuous Extension}\label{sec:row}

In this section, we present an alternative continuous extension of the EFX problem, which enables us to establish theoretical results regarding the existence of an EFX allocation through nonlinear optimization and fixed-point theorems.  

Given an arbitrary (fractional) feasible solution $$x \in P= \Big\{ x \in [0,1]^{m \times n} : \sum_{i=1}^n x_{\ell i} = 1 \ \forall \ell \in [m] \Big\},$$ suppose that we round each row of \( x \) independently to a basis vector. That is, we let \( X \in P \) be an \emph{integral} allocation such that, independently for each row \( \ell\in [m] \), we sample one of its elements according to the probability distribution \( (x_{\ell 1}, \ldots, x_{\ell n}) \) and set it to 1. All other coordinates in that row are set to zero. Since the rounding is done independently across rows, it is easy to see that for each column \( X_i \) of the random binary matrix \( X = (X_1, \ldots, X_n) \), each entry \( X_{\ell i} \), for \( \ell \in [m] \), is set to 1 independently with probability \( x_{\ell i} \).

Next, we analyze the expectation $\mathbb{E}_{X \sim x} \left[ \max_{i \neq j} u_{ij}(X) \right]$ under the row-wise randomized rounding scheme described above.\footnote{Here, $\mathbb{E}_{X \sim x}[\cdot]$ denotes expectation with respect to a random binary matrix $X$, obtained by applying row-wise independent rounding to $x$.
}In particular, to find a solution \( x \) that minimizes this expected value, we first derive an upper bound for $\mathbb{E}_{X \sim x} \left[ \max_{i \neq j} u_{ij}(X) \right]$, and then minimize this upper bound over all feasible fractional allocations \( x \in P \). To derive such an upper bound, note that for any \( \lambda > 0 \), we can write:

\begin{align}\label{eq:simple-max-trick}
\mathbb{E}[\max_{i\neq j}u_{ij}(X)]&= \frac{1}{\lambda}\ln \left(e^{\lambda\mathbb{E}[\max_{i\neq j}u_{ij}(X)]}\right)\leq \frac{1}{\lambda}\ln \left(\mathbb{E}[e^{\lambda \max_{i\neq j} u_{ij}(X)}]\right)\cr
&=\frac{1}{\lambda}\ln \left(\mathbb{E}[ \max_{i\neq j} e^{\lambda u_{ij}(X)}]\right)\leq \frac{1}{\lambda}\ln \Big(\mathbb{E}[\sum_{i\neq j}e^{\lambda u_{ij}(X)}]\Big)\cr
&= \frac{1}{\lambda}\ln \Big(\sum_{i\neq j}\mathbb{E}[e^{\lambda u_{ij}(X)}]\Big),
\end{align} 
where the first inequality uses Jensen's inequality. Thus, we only need to upper-bound $\mathbb{E}[e^{\lambda u_{ij}(X)}]$ for an arbitrary pair of $i\neq j$, which is what we are going to do next. 

%&=\frac{2\ln(n)+\max_{i,j}\ln(\mathbb{E}[e^{\lambda u_{ij}(X)}])}{\lambda},

Using the the tower property of the conditional expectation, we have  
%Therefore, by the linearity of the expectation and the definition of the multilinear extension, we have $\mathbb{E}[u_{ij}(X)]=\mathbb{E}[f_i(X_j)]-\mathbb{E}[v_i(X_i)]=f^M_i(x^*_i)-v^M_i(x^*_i)=u^M_{ij}(x^*)\leq f^*<0$. Our goal is to find an upper bound for $\mathbb{E}[f(X)]=\mathbb{E}[\max_{i,j}u_{ij}(X)]$. Let $\lambda>0$ be an arbitrary scalar. Then,
\begin{align}\label{eq:tower}
\mathbb{E}[e^{\lambda u_{ij}(X)}]=\mathbb{E}\big[e^{\lambda f_i(X_j)}\mathbb{E}[e^{-\lambda v_i(X_i)}|X_j]\big].
\end{align}
Since the rows of \( x \) are rounded independently, \( X_{\ell i} \) only depends on \( X_{\ell j} \), and is independent of \( X_{\ell'i} \) and \( X_{\ell'j} \) for any \( \ell' \neq \ell \). Therefore, using the definition of conditional expectation, we can write:
\begin{align}\label{eq:product-condition-expectation}
\mathbb{E}[e^{-\lambda v_i(X_i)}|X_j]=\mathbb{E}[\prod_{\ell=1}^m e^{-\lambda v_{\ell i}X_{\ell i}}|X_j]=\prod_{\ell=1}^m\mathbb{E}[ e^{-\lambda v_{\ell i}X_{\ell i}}|X_{\ell j}].
\end{align}
Moreover, since exactly one entry in each row is set to 1, we can compute each term $\mathbb{E}[ e^{-\lambda v_{\ell i}X_{\ell i}}|X_{\ell j}]$ in the above product as
\begin{align}\label{eq:exp-two}
\mathbb{E}[ e^{-\lambda v_{\ell i}X_{\ell i}}|X_{\ell j}=1]&=1,\cr
\mathbb{E}[ e^{-\lambda v_{\ell i}X_{\ell i}}|X_{\ell j}=0]&=1\times \mathbb{P}\{X_{\ell i}=0|X_{\ell j}=0\}+e^{-\lambda v_{\ell i}}\times \mathbb{P}\{X_{\ell i}=1|X_{\ell j}=0\}\cr
&=1\times \frac{1-x_{\ell i}-x_{\ell j}}{1-x_{\ell j}}+e^{-\lambda v_{\ell i}}\times \frac{x_{\ell i}}{1-x_{\ell j}}\cr
&=1-(1-e^{-\lambda v_{\ell i}})\frac{x_{\ell i}}{1-x_{\ell j}}.
\end{align}
Let us define
\[
w_{ij\ell} := \left(1 - e^{-\lambda v_{\ell i}} \right) \frac{x_{\ell i}}{1 - x_{\ell j}},
\]
and note that for any \( \lambda > 0 \), we have \( w_{ij\ell} \in [0, 1] \). Thus, we can express both relations in~\eqref{eq:exp-two} using a single equation as follows:
\begin{align}\nonumber
\mathbb{E}[ e^{-\lambda v_{\ell i}X_{\ell i}}|X_{\ell j}]=1-w_{ij\ell}(1-X_{\ell j}).
\end{align}
Therefore, using \eqref{eq:product-condition-expectation}, we have
\begin{align}\nonumber
\mathbb{E}[e^{-\lambda v_{i}(X_i)}|X_j]=\prod_{\ell=1}^m\Big(1- w_{ij\ell}(1-X_{\ell j})\Big).
\end{align} 
Substituting this relation into \eqref{eq:tower}, we obtain
\begin{align}\nonumber
\mathbb{E}[e^{\lambda u_{ij}(X)}]= \mathbb{E}\Big[e^{\lambda f_i(X_j)}\prod_{\ell=1}^m\Big(1- w_{ij\ell}(1-X_{\ell j})\Big)\Big].
\end{align}
Next, we proceed to upper-bound the term $e^{\lambda f_i(X_j)}$. Using the definition of $f_i(X_j)$ from \eqref{eq:fi-def}, we have
\begin{align}\nonumber
e^{\lambda f_i(X_j)}&=e^{\lambda \sum_{\ell} v_{\ell i}X_{\ell j}}\times e^{-\lambda \min_{k\in X_j}v_{ki}}\cr
&=e^{\lambda \sum_{\ell} v_{\ell i}X_{\ell j}}\times \max_{k\in X_j}e^{-\lambda v_{ki}}\cr
&\leq e^{\lambda \sum_{\ell} v_{\ell i}X_{\ell j}}\times \sum_{k}X_{kj}e^{-\lambda v_{ki}}\cr
&=\sum_{k}X_{kj}e^{\big(\lambda \sum_{\ell} v_{\ell i}X_{\ell j}-\lambda v_{ki}\big)}.
\end{align}
Substituting this relation into the former and using the linearity of expectation, we obtain
\begin{align}\label{e-l-u}
\mathbb{E}[e^{\lambda u_{ij}(X)}]&\leq \mathbb{E}\Big[\sum_{k}X_{kj}e^{\big(\lambda \sum_{\ell} v_{\ell i}X_{\ell j}-\lambda v_{ki}\big)}\prod_{\ell=1}^m\Big(1- w_{ij\ell}(1-X_{\ell j})\Big)\Big]\cr
&=\sum_{k}\mathbb{E}\Big[X_{kj}e^{-\lambda v_{ki}}\prod_{\ell=1}^m e^{\lambda v_{\ell i}X_{\ell j}}\Big(1- w_{ij\ell}(1-X_{\ell j})\Big)\Big].
\end{align}
Since $X_{kj}$ is independent of $X_{k'j}$ for any $k\neq k'$, the expectation in \eqref{e-l-u} can be computed as follows:
\begin{align}\label{eq:residual}
\mathbb{E}\Big[X_{kj}e^{-\lambda v_{ki}}\prod_{\ell=1}^m e^{\lambda v_{\ell i}X_{\ell j}}\Big(1&- w_{ij\ell}(1-X_{\ell j})\Big)\Big]\cr
&=x_{kj}e^{-\lambda v_{ki}}\mathbb{E}\Big[\prod_{\ell=1}^m e^{\lambda v_{\ell i}X_{\ell j}}\Big(1- w_{ij\ell}(1-X_{\ell j})\Big)\Big|X_{kj}=1\Big]\cr 
&=x_{kj}\mathbb{E}\Big[\prod_{\ell\neq k}e^{\lambda v_{\ell i}X_{\ell j}}\Big(1- w_{ij\ell}(1-X_{\ell j})\Big)\Big|X_{kj}=1\Big]\cr
&=x_{kj}\prod_{\ell\neq k}\mathbb{E}\Big[e^{\lambda v_{\ell i}X_{\ell j}}\Big(1- w_{ij\ell}(1-X_{\ell j})\Big)\Big]\cr
&=x_{kj}\prod_{\ell\neq k}\Big((1-x_{\ell j})(1-w_{ij\ell})+x_{\ell j}e^{\lambda v_{\ell i}}\Big)\cr
&=x_{kj}\prod_{\ell\neq k}\Big(1-x_{\ell i}-x_{\ell j}+x_{\ell i}e^{-\lambda v_{\ell i}}+x_{\ell j}e^{\lambda v_{\ell i}}\Big),
\end{align}
where in the last equality we have used the definition of $w_{ij\ell}$. Substituting \eqref{eq:residual} into \eqref{e-l-u}, we get
\begin{align}\nonumber
\mathbb{E}[e^{\lambda u_{ij}(X)}]&\leq \sum_{k} x_{kj}\prod_{\ell\neq k}\Big(1-x_{\ell i}-x_{\ell j}+x_{\ell i}e^{-\lambda v_{\ell i}}+x_{\ell j}e^{\lambda v_{\ell i}}\Big).
\end{align}
Finally, by replacing this relation into \eqref{eq:simple-max-trick}, for any $\lambda>0$ and $x\in P$, we obtain
\begin{align}\label{eq:final-g}
\mathbb{E}_{X\sim x}[\max_{i\neq j}u_{ij}(X)]\leq \frac{1}{\lambda}\ln \left(\sum_{i\neq j,k}x_{kj}\prod_{\ell\neq k}\Big(1-x_{\ell i}-x_{\ell j}+x_{\ell i}e^{-\lambda v_{\ell i}}+x_{\ell j}e^{\lambda v_{\ell i}}\Big)\right),
\end{align}
which provides a closed-form  bound for the quantity of interest $\mathbb{E}_{X\sim x}[\max_{i\neq j}u_{ij}(X)]$. The following lemma presents an equivalent characterization of EFX allocations using the continuous-extension upper bound obtained in~\eqref{eq:final-g}.

\begin{lemma}\label{lemm:g}
Let $P= \left\{ x\in \mathbb{R}_+^{m n} : \sum_{i=1}^n x_{\ell i} = 1 \ \forall \ell \in [m] \right\}$ be the partition polytope, and consider the multivariate function $g:P\times (0,\infty)\to \mathbb{R}$ defined by
\begin{align}\nonumber
g(x,\lambda)=\frac{1}{\lambda}\ln \left(\sum_{i\neq j,k}x_{kj}\prod_{\ell\neq k}\Big(1-x_{\ell i}-x_{\ell j}+x_{\ell i}e^{-\lambda v_{\ell i}}+x_{\ell j}e^{\lambda v_{\ell i}}\Big)\right).
\end{align}  
Then, an EFX allocation exists if and only if $\inf\{g(\lambda,x):\lambda>0, x\in P\}\leq 0.$
\end{lemma}
\begin{proof}
First, let us assume $\inf\{g(x,\lambda):\lambda>0, x\in P\}\leq 0$. As $g(x,\lambda)$ is a continuous function over a convex domain, for any $\delta>0$, there exist $x^*\in P, \lambda^*>0$, such that $g(x^*,\lambda^*)< \delta$. Since \eqref{eq:final-g} holds for any $x\in P$ and $\lambda>0$, we have 
\begin{align*}
\mathbb{E}_{X\sim x^*}[\max_{i\neq j}u_{ij}(X)]\leq g(x^*,\lambda^*)<\delta,
\end{align*}
which shows that there exists at least one integral allocation \( X^* \in P \), obtained via independent row-wise rounding of \( x^* \), such that $\max_{i \neq j} u_{ij}(X^*) < \delta$. Since the above argument holds for any \( \delta > 0 \), and there are only finitely many integral allocations (and hence, \( \max_{i \neq j} u_{ij}(X) \), for \( X \in P \cap \{0,1\}^{mn} \), can take only finitely many distinct values), it follows that for a sufficiently small \( \delta > 0 \), there exists at least one integral allocation \( X^* \in P \) such that $
\max_{i \neq j} u_{ij}(X^*) \leq 0.$ Therefore, \( X^* \) must be an EFX allocation.

Conversely, suppose that there exists an (integral) EFX allocation \( X^* \in P \), such that $\max_{i \neq j} u_{ij}(X^*) \leq 0$. By taking \( x = X^* \), a simple calculation shows that for any \( \lambda \in (0, \infty) \):
\begin{align}\label{eq:lemms-second-part}
g(X^*,\lambda)&=\frac{1}{\lambda}\ln\Bigg( \sum_{i\neq j}\sum_{k\in X^*_j}\exp\Big(\lambda\big(\sum_{\ell\in X^*_j\setminus \{k\}} v_{\ell i}-\sum_{\ell\in X^*_i}v_{\ell i}\big)\Big)\Bigg)\cr
&\leq \frac{1}{\lambda}\ln \Bigg(\sum_{i\neq j}|X^*_j|\exp\Big(\lambda u_{ij} (X^*)\Big)\Bigg)\cr
&\leq \frac{1}{\lambda}\ln \Bigg(nm \exp\Big(\lambda \max_{i\neq j}u_{ij} (X^*)\Big)\Bigg)\cr
&=\frac{\ln (nm)}{\lambda}+\max_{i\neq j}u_{i,j} (X^*)\cr
&\leq \frac{\ln (nm)}{\lambda},
\end{align}  
where the first equality holds because $x_{kj}=1$ implies that $k\in X^*_j$ and $k\notin X^*_i$. Moreover, the first inequality holds because $u_{ij}(X^*)=f_i(X^*_j)-v_i(X^*_i)$, and 
\begin{align*}
\sum_{\ell\in X^*_j\setminus \{k\}} v_{\ell i}\leq \sum_{\ell\in X^*_j} v_{\ell i}-\min_{\ell\in X^*_j}v_{\ell i}=f_i(X^*_j)\ \ \forall k\in X^*_j.  
\end{align*}
Therefore, by letting $\lambda\to \infty$, the right-hand side of the inequality \eqref{eq:lemms-second-part} goes to zero, which shows that $\inf\{g(x,\lambda): x\in P, \lambda>0\}\leq 0$. 
\end{proof}

\medskip
\begin{remark}
The advantage of Lemma \ref{lemm:g} is that (i) it reduces the problem of finding an EFX allocation to solving a continuous and smooth nonlinear optimization problem over a polyhedron, and (ii) it does not require an integrality constraint; instead, an integral solution is obtained from an optimal fractional one using row-wise independent randomized rounding.   
\end{remark}

\medskip
\begin{remark}\label{rem:large-lambda}
If $\inf\{g(x,\lambda) : \lambda > 0,\, x \in P\} \leq 0$, then for any EFX allocation $X^*$ (whose existence is guaranteed by Lemma~\ref{lemm:g}), we have $\limsup_{\lambda \to \infty} g(X^*, \lambda) \leq 0$. This follows directly from~\eqref{eq:lemms-second-part}, which shows that for any EFX allocation $X^*$ and any $\lambda > 0$, we have $g(X^*, \lambda) \leq \frac{\ln(nm)}{\lambda}$.
\end{remark}

\medskip
While Lemma~\ref{lemm:g} establishes a connection between EFX allocations and the optimal value of a nonlinear program, the objective function $g(x, \lambda)$ unfortunately lacks a succinct representation and involves $O(mn^{m+1})$ terms. Therefore, in the following theorem, we show that the same result holds for a much simpler nonlinear function with at most $O(nm)$ terms. This simplified function is obtained through a suitable change of variables in $g(x, \lambda)$ and is analyzed in the asymptotic regime as $\lambda \to \infty$.

\medskip
\begin{theorem}\label{thm:f-y}
Consider the continuous function $f:\mathbb{R}^{mn}\to \mathbb{R}$, which is defined by
\begin{align}\nonumber 
f(y)=\max_{i\neq j, k}\Big(y_{kj}+\sum_{\ell\neq k}\max\Big\{y_{\ell i}-v_{\ell i}, y_{\ell j}+v_{\ell i}, \max_{r\neq i,j }y_{\ell r}\Big\}\Big)-\sum_{\ell}\max_{r} y_{\ell r}. 
\end{align}
Then, an EFX allocation exists if and only if $\inf_{y\in \mathbb{R}^{mn}} f(y)\leq 0$.
\end{theorem}
\begin{proof}
Using Lemma~\ref{lemm:g}, it suffices to show that $\inf\{g(x, \lambda) : \lambda > 0,\ x \in P\} \leq 0$ if and only if $\inf_{y \in \mathbb{R}^{mn}} f(y) \leq 0$. Given any $(x, \lambda) \in \big(P \cap (0,1)^{mn}\big) \times (0, \infty)$, one can find a $y \in \mathbb{R}^{mn}$ such that
\footnote{It is enough to take $y_{\ell i} = \frac{\ln x_{\ell i}}{\lambda}$, and this choice is unique up to a constant shift in each row.}
\begin{align}\label{eq:x-y-exp}
x_{\ell i}=\frac{e^{\lambda y_{\ell i}}}{\sum_{j\in N}e^{\lambda y_{\ell j}}}\ \  \forall \ell \in M, i\in N.
\end{align} 
Let us denote the denominator in \eqref{eq:x-y-exp} by $Z_\ell=\sum_{j\in N}e^{\lambda y_{\ell j}}.$ Using this change of variable, we have
\begin{align}\label{eq:g-x-y}
g(x,\lambda)&=\frac{1}{\lambda}\ln \left(\sum_{i\neq j,k}x_{kj}\prod_{\ell\neq k}\Big(1-x_{\ell i}-x_{\ell j}+x_{\ell i}e^{-\lambda v_{\ell i}}+x_{\ell j}e^{\lambda v_{\ell i}}\Big)\right)\cr
&=\frac{1}{\lambda}\ln \left(\sum_{i\neq j,k}\frac{e^{\lambda y_{kj}}}{Z_{k}}\prod_{\ell\neq k}\Big(1-\frac{e^{\lambda y_{\ell i}}}{Z_{\ell}}-\frac{e^{\lambda x_{\ell j}}}{Z_{\ell}}+\frac{e^{\lambda(y_{\ell i}-v_{\ell i})}}{Z_{\ell}}+\frac{e^{\lambda (y_{\ell j}+v_{\ell i})}}{Z_{\ell}}\Big)\right)\cr
&=\frac{1}{\lambda}\ln \left(\sum_{i\neq j,k}\frac{e^{\lambda y_{kj}}}{\prod_{\ell} Z_{\ell}}\prod_{\ell\neq k}\Big(Z_{\ell}-e^{\lambda y_{\ell i}}-e^{\lambda x_{\ell j}}+e^{\lambda(y_{\ell i}-v_{\ell i})}+e^{\lambda (y_{\ell j}+v_{\ell i})}\Big)\right)\cr
&=-\frac{1}{\lambda}\sum_{\ell}\ln Z_{\ell}+\frac{1}{\lambda}\ln \left(\sum_{i\neq j,k}e^{\lambda y_{kj}}\prod_{\ell\neq k}\Big(\sum_{r\neq i,j}e^{\lambda y_{\ell r}}+e^{\lambda(y_{\ell i}-v_{\ell i})}+e^{\lambda (y_{\ell j}+v_{\ell i})}\Big)\right).
\end{align}
Next, we analyze the limit of expression~\eqref{eq:g-x-y} as $\lambda \to \infty$, and show that the limit always exists and equals $f(y)$.

Since $g(x, \lambda)$ is differentiable with respect to $\lambda$ over $(0, \infty)$, we can apply Hôpital's rule. The limit of the first term in~\eqref{eq:g-x-y} as $\lambda \to \infty$ is given by 
\begin{align}\label{eq:first-term-lim}
\lim_{\lambda\to \infty} -\frac{1}{\lambda}\sum_{\ell}\ln Z_{\ell}&=-\sum_{\ell} \lim_{\lambda\to \infty} \frac{1}{\lambda}\ln \Big(\sum_{j}e^{\lambda y_{\ell j}}\Big)\cr
&=-\sum_{\ell}\lim_{\lambda\to \infty}\sum_{j} y_{\ell j}\frac{e^{\lambda y_{\ell j}}}{\sum_{r}e^{\lambda y_{\ell r}}}\cr
&=-\sum_{\ell}\max_{j}y_{\ell j}. 
\end{align}

To compute the limit of the second term in \eqref{eq:g-x-y}, we note that by expanding the product, the argument inside the logarithm can be written as a sum of $mn^{m+1}$ exponential terms of the form $\sum_{d=1}^{mn^{m+1}} e^{\lambda a_d}$. Therefore, similar to the first term, the limit of $\frac{1}{\lambda} \ln\left(\sum_{d} e^{\lambda a_d}\right)$ as $\lambda \to \infty$ is given by $\max_d a_d$. Unfortunately, this does not yield a succinct expression, as the maximum is taken over $mn^{m+1}$ different values. However, since we know that the limit exists and is equal to $\max_d a_d$, we can provide an alternative and more succinct expression for $\max_d a_d$. Observe that each term inside the product $\prod_{\ell \neq k} \left( \sum_{r \neq i,j} e^{\lambda y_{\ell r}} + e^{\lambda(y_{\ell i} - v_{\ell i})} + e^{\lambda(y_{\ell j} + v_{\ell i})} \right)$ involves disjoint row variables $y_{\ell}=(y_{\ell 1},\ldots,y_{\ell n})$, and the maximum exponent resulting from expanding the product is obtained by summing the maximum exponents from each individual term, that is, $\sum_{\ell \neq k} \max \left\{ y_{\ell i} - v_{\ell i},\ y_{\ell j} + v_{\ell i},\ \max_{r \neq i,j} y_{\ell r} \right\}$. Finally, taking into account the last exponential term $e^{\lambda y_{k j}}$ and maximizing over all indices $i \neq j$ and $k$ (due to the outer summation), we conclude that the limit of the second term in \eqref{eq:g-x-y} as $\lambda \to \infty$ equals
\begin{align*}
\max_{d}a_d=\max_{i\neq j, k}\Big(y_{kj}+\sum_{\ell\neq k}\max\Big\{y_{\ell i}-v_{\ell i}, y_{\ell j}+v_{\ell i}, \max_{r\neq i,j}y_{\ell r}\Big\}\Big).
\end{align*}
This relation together with \eqref{eq:first-term-lim} and \eqref{eq:g-x-y} shows that $\lim_{\lambda\to \infty}g(x,\lambda)=f(y)$.         

Now, let us assume that $\inf\{ g(x, \lambda) : x \in P,\, \lambda > 0 \} \leq 0$, which, according to Lemma~\ref{lemm:g}, implies that an EFX allocation $x^*$ exists. Let $M = 2V + 1$, where $V = \sum_{\ell, i} v_{\ell i}$, and define $y^*$ as
\begin{align}\nonumber
y^*_{\ell i}=\begin{cases}
0 &\mbox{if} \ x^*_{\ell i}=1,\\
-M &\mbox{if} \ x^*_{\ell i}=0.
\end{cases}
\end{align} 
Moreover, let $x$ be defined using \eqref{eq:x-y-exp} for the pair $(y^*, \lambda)$. Since we have shown that $\lim_{\lambda \to \infty} g(x, \lambda) = f(y^*)$, for any $\delta > 0$, there exists $\lambda_\delta > 0$ such that $|f(y^*) - g(x, \lambda)| < \delta$ for all $\lambda > \lambda_\delta$. Furthermore, using Remark~\ref{rem:large-lambda}, there exists an arbitrarily large $\lambda^* > \max\left\{ \lambda_\delta, \ln \frac{1}{\delta}, 2m^2n^3 \right\}$ such that $g(x^*, \lambda^*) < \delta$. Thus, if we let $\tilde{x}$ be defined using \eqref{eq:x-y-exp} for the pair $(y^*, \lambda^*)$ (i.e., $\tilde{x}$ is the same as $x$ for the specific choice $\lambda = \lambda^*$), we can write
\begin{align}\label{eq:x-tilde x}
f(y^*)-g(x^*,\lambda^*)&= \big(f(y^*)-g(\tilde{x},\lambda^*)\big)+\big(g(\tilde{x},\lambda^*)-g(x^*,\lambda^*)\big)\cr
&< \delta+\big(g(\tilde{x},\lambda^*)-g(x^*,\lambda^*)\big). 
\end{align}  
We claim that the second term in \eqref{eq:x-tilde x} is at most $\delta.$ To that end, let us define
\begin{align}\label{eq:epsilon}
\epsilon&=\max_{i\neq j, \ell}\Big|\big(1-\tilde{x}_{\ell i}-\tilde{x}_{\ell j}+\tilde{x}_{\ell i}e^{-\lambda^* v_{\ell i}}+\tilde{x}_{\ell j}e^{\lambda^* v_{\ell i}}\big)-\big(1-x^*_{\ell i}-x^*_{\ell j}+x^*_{\ell i}e^{-\lambda^* v_{\ell i}}+x^*_{\ell j}e^{\lambda^* v_{\ell i}}\big)\Big|\cr
&\qquad\leq 4\max_{\ell, i}|\tilde{x}_{\ell i}-x^*_{\ell i}|\cdot \max_{\ell, i}e^{\lambda^*v_{\ell i}}\cr
&\qquad< 4ne^{-\lambda^* M}\cdot e^{\lambda^*\max_{\ell, i} v_{\ell i}}\cr
&\qquad\leq 4ne^{-\lambda^* (M-V)},
\end{align} 
where the second inequality holds because, using \eqref{eq:x-y-exp} and the definition of $y^*$, for any $\ell$ and $i$, we have
\begin{align*}
|\tilde{x}_{\ell i}-x^*_{\ell i}|&\leq \max\Big\{\Big|\frac{e^{-\lambda^* M}}{\sum_j e^{\lambda^* y^*_{\ell j}}}-0\Big|, \Big|\frac{1}{\sum_j e^{\lambda^* y^*_{\ell j}}}-1\Big|\Big\}\cr
&= \max\Big\{\frac{e^{-\lambda^* M}}{1+(n-1)e^{-\lambda^* M}},1-\frac{1}{1+(n-1)e^{-\lambda^* M}}\Big\}\cr
&< ne^{-\lambda^* M}. 
\end{align*}
Now, let us define $$Q(z)=\sum_{i\neq j,k}z_{kj}\prod_{\ell\neq k}\Big(1-z_{\ell i}-z_{\ell j}+z_{\ell i}e^{-\lambda^* v_{\ell i}}+z_{\ell j}e^{\lambda^* v_{\ell i}}\Big).$$ 
Then, we can write
\begin{align}\nonumber
Q(\tilde{x})&-Q(x^*)\leq \sum_{\substack{i\neq j,k\\ x^*_{kj}=0}}\tilde{x}_{kj}\prod_{\ell\neq k}\Big(1-\tilde{x}_{\ell i}-\tilde{x}_{\ell j}+\tilde{x}_{\ell i}e^{-\lambda^* v_{\ell i}}+\tilde{x}_{\ell j}e^{\lambda^* v_{\ell i}}\Big)\cr
&+\sum_{\substack{i\neq j,k\\ x^*_{kj}=1}}\Bigg(\!\prod_{\ell\neq k}\Big(1-\tilde{x}_{\ell i}-\tilde{x}_{\ell j}+\tilde{x}_{\ell i}e^{-\lambda^* v_{\ell i}}+\tilde{x}_{\ell j}e^{\lambda^* v_{\ell i}}\Big)\!-\prod_{\ell\neq k}\Big(1-x^*_{\ell i}-x^*_{\ell j}+x^*_{\ell i}e^{-\lambda^* v_{\ell i}}+x^*_{\ell j}e^{\lambda^* v_{\ell i}}\Big)\!\Bigg)\cr
&\leq \sum_{\substack{i\neq j,k\\ x^*_{kj}=0}}\tilde{x}_{kj}\prod_{\ell\neq k}e^{\lambda^* v_{\ell i}}+\sum_{\substack{i\neq j,k\\ x^*_{kj}=1}}(m-2)\epsilon\prod_{\ell\neq k}(e^{\lambda^* v_{\ell i}}+\epsilon)\cr
&\leq mn^2 \exp\Big(-\lambda^*(M-\max_i\sum_{\ell\neq k}v_{\ell i})\Big)+m^2n^3\exp\Big(-\lambda^*(M-\max_i\sum_{\ell\neq k}v_{\ell i})\Big)\cr
&\leq 2m^2n^3e^{-\lambda^*(M-V)}.
\end{align}
In the above derivations, the first inequality follows because the product terms are nonnegative and $\tilde{x}_{kj}\leq 1$. The second inequality holds by noting that the $\ell$th term in the product is at most $e^{\lambda^* v_{\ell i}}$, and by using the definition of $\epsilon$ along with the Mean Value Theorem to upper-bound the difference of two products. Finally, the third inequality uses the fact that $\tilde{x}_{kj} \leq e^{-\lambda^* M}$ whenever $x^*_{\ell j} = 0$, and also uses the upper bound on $\epsilon$ from \eqref{eq:epsilon}. On the other hand, from \eqref{eq:lemms-second-part} we know that 
\begin{align*}
Q(x^*)=\sum_{\substack{i\neq j,k \\ x^*_{kj}=1}}\exp\Big(\lambda^*\big(\sum_{\ell\neq k, x^*_{\ell j}=1} v_{\ell i}-\sum_{\ell: x^*_{\ell i}=1}v_{\ell i}\big)\Big)\ge e^{-\lambda^*V}.
\end{align*}
Thus, we can write   
\begin{align}\nonumber
g(\tilde{x},\lambda^*)-g(x^*,\lambda^*)&=\frac{1}{\lambda^*}\ln \frac{Q(\tilde{x})}{Q(x^*)}\cr
&\leq \frac{1}{\lambda^*}\ln \Big(1+\frac{2m^2n^3e^{-\lambda^*(M-V)}}{e^{-\lambda^* V}}\Big)\cr
&\leq \frac{2m^2n^3}{\lambda^*}e^{-\lambda^*(M-2V)}\cr
&\leq e^{-\lambda^*}<\delta.
\end{align}
Using this relation in \eqref{eq:x-tilde x}, we have shown that for any $\delta > 0$, there exist $y^* \in \mathbb{R}^{mn}$ and a pair $(x^*, \lambda^*)$ such that $g(x^*, \lambda^*) \leq \delta$ and $f(y^*) - g(x^*, \lambda^*) < 2\delta$. Thus, for any $\delta > 0$, there exists $y^* \in \mathbb{R}^{mn}$ such that $f(y^*) < 3\delta$. This shows that $\inf_{y \in \mathbb{R}^{mn}} f(y) \leq 0$.

Conversely, if $\inf_{y \in \mathbb{R}^{mn}} f(y) \leq 0$, then for any $\delta > 0$, there exists $y^* \in \mathbb{R}^{mn}$ such that $f(y^*) < \delta$. Define $x^* \in P$ using \eqref{eq:x-y-exp} for the pair $(y^*, \lambda)$, where $\lambda > 0$ is free to choose. Since we have already shown that $\lim_{\lambda \to \infty} g(x^*, \lambda) = f(y^*)$, we get 
\begin{align*}
\inf_{x\in P,\lambda>0} g(x,\lambda)\leq \lim_{\lambda \to \infty} g(x^*,\lambda)=f(y^*)< \delta.
\end{align*} 
As this argument holds for any $\delta > 0$, we obtain $\inf_{x \in P,\, \lambda > 0} g(x, \lambda) \leq 0$, which completes the proof.   
\end{proof}

We note that Theorem \ref{thm:f-y} provides a necessary and sufficient condition for the existence of an EFX allocation in terms of the sign of the optimal value of an unconstrained optimization problem. In the next section, we leverage this objective function to establish an algorithmic result for obtaining an EFX allocation (if it exists).

\section{Difference of Convex Programming for Finding an EFX Allocation}\label{sec:DC}

As we showed in Theorem \ref{thm:f-y}, the existence of an EFX allocation is equivalent to checking whether $\inf_{y \in \mathbb{R}^{mn}} f(y) \leq 0$. In this section, we provide a framework based on \emph{difference-of-convex} (DC) optimization to address this problem. More specifically, to determine whether the unconstrained optimization $\inf_{y \in \mathbb{R}^{mn}} f(y)$ has a nonpositive optimal objective value, we observe that $f(y)$ can be written as $f(y) = \bar{h}(y) - h(y)$, where
\begin{align}\label{eq:two-convex}
&\bar{h}(y):=\max_{i\neq j, k}\Big(y_{kj}+\sum_{\ell\neq k}\max\Big\{y_{\ell i}-v_{\ell i}, y_{\ell j}+v_{\ell i}, \max_{r\neq i,j }y_{\ell r}\Big\}\Big),\cr
&h(y):= \sum_{\ell}\max_{r} y_{\ell r}.
\end{align}
In particular, both $\bar{h}(y)$ and $h(y)$ are continuous and convex functions of $y$. This allows us to use the DC optimization framework to study $\inf_{y \in \mathbb{R}^{mn}} \bar{h}(y) - h(y)$.

\begin{remark}
While determining whether the optimal value is nonpositive is an NP-complete problem for the difference of two general convex functions, in our setting, the convex functions take the special form given in \eqref{eq:two-convex}, which may simplify the analysis. This reduction, offers a new direction for studying the complexity of finding an EFX allocation—an avenue we leave for future research. 
\end{remark}

Next, we proceed to reformulate $\inf_{y \in \mathbb{R}^{mn}} \bar{h}(y) - h(y)$ as a concave minimization problem over a polytope. First, we note that the objective function $f(y)$ can be rewritten as
\begin{align*}
f(y)&=\bar{h}(y)-h(y)=\max_{i\neq j, k}\Big(y_{kj}+\sum_{\ell\neq k}\max\Big\{y_{\ell i}-v_{\ell i}, y_{\ell j}+v_{\ell i}, \max_{r\neq i,j }y_{\ell r}\Big\}\Big)-h(y).    
\end{align*} 
Let us introduce auxiliary variables $w$ and $z_{\ell i j }$, which aim to represent $\bar{h}(y)$ and $\max\{y_{\ell i}+v_{\ell i}, y_{\ell j}-v_{\ell i}, \max_{r\neq i,j} y_{\ell r}\}$, respectively. Then, the optimization problem $\inf_{y\in \mathbb{R}^{nm}}f(y)$ can be written as
\begin{align}\label{eq:DC-program}
\text{minimize} \quad & w-h(y) \cr
\text{subject to} \quad 
& y_{kj}+\sum_{\ell\neq k}z_{\ell i j}-w\le 0 \quad \forall k,i\neq j \cr
&y_{\ell i}-z_{\ell i j}\le -v_{\ell i} \quad \forall \ell,i\neq j \cr
&y_{\ell j}-z_{\ell i j}\le v_{\ell i} \quad \forall \ell,i\neq j \cr
&y_{\ell r}-z_{\ell i j}\le 0 \quad \forall \ell,i\neq j, r\neq i,j.
\end{align}
It is worth noting that, since the objective function $w - h(y) = w - \sum_{\ell} \max_r y_{\ell r}$ is a (piecewise linear) concave function, the optimal solution occurs at one of the extreme points of the feasible polytope.

To find a stationary point of such a problem, a widely used method is the DC Algorithm (DCA) \cite{lethi1997dca}, an iterative scheme that approximates the non-convex part of the objective by solving a sequence of convex subproblems. Specifically, starting from an initial point, at iteration $t$, the DCA computes a subgradient of $h$ at $y^{(t)}$, i.e., $v^{(t)} \in \partial h(y^{(t)})$, and solves the convex subproblem $$y^{(t+1)} = \arg\min_{y \in Y} \left\{ \bar{h}(y) - \langle v^{(t)}, y \rangle \right\},$$ where $Y$ is the feasible polytope. Adapting this algorithm to the program \eqref{eq:DC-program}, we obtain Algorithm 1.

\begin{algorithm}\label{alg:DCA}
\caption{Difference of Convex Program Algorithm (DCA)}
\begin{algorithmic}[1]
\State \textbf{Initialize:} Choose an initial point \( y^0 \in \mathbb{R}^{nm} \), an error tolerance $\delta\ge 0$, and set \( t = 0 \).
\While{not converged}
    \State Let $L^t(y)=\sum_{\ell}y_{\ell r^t_{\ell}}$, where $r^t_{\ell}=\argmax_{r}y^t_{\ell r}$.
    \State Solve the following LP, and let $y^{t+1}$ be the optimal solution to the $y$ variable: 
\begin{align*}
\text{minimize} \quad & w-L^t(y) \\
\text{subject to} \quad 
& y_{kj}+\sum_{\ell\neq k}z_{\ell i j}-w\le 0 \quad \forall k,i\neq j \cr
&y_{\ell i}-z_{\ell i j}\le -v_{\ell i} \quad \forall \ell,i\neq j \cr
&y_{\ell j}-z_{\ell i j}\le v_{\ell i} \quad \forall \ell,i\neq j \cr
&y_{\ell r}-z_{\ell i j}\le 0 \quad \forall \ell,i\neq j, r\neq i,j.
\end{align*}
    \State Check convergence: if \( \| y^{t+1} - y^{t} \|\leq \delta\), \textbf{stop}, and return $y^{t+1}$ and the optimal value to the LP.
    \State Update: \( t \gets t + 1 \)
\EndWhile
\end{algorithmic}
\end{algorithm}

As the tolerance parameter $\delta$ in Algorithm~1 approaches $0$, the iterates converge to a stationary point of the program~\eqref{eq:DC-program}. In particular, if the value returned by the algorithm is nonpositive, the resulting $y$ solution corresponds to an EFX allocation (through the change of variables in~\eqref{eq:x-y-exp} as $\lambda \to \infty$). Of course, solving a DC program is generally NP-hard, and this may also be the case for solving~\eqref{eq:DC-program}. Therefore, although there is no guarantee that a globally optimal solution will be achieved, this approach nonetheless provides a systematic method for searching for EFX allocations, if they exist.

\section{A Fixed Point Formulation for the Existence of an EFX Allocation}\label{sec:fixed}

In this section, we provide an alternative fixed-point reformulation for the existence of EFX allocations. More specifically, the optimization condition for the existence of an EFX allocation, as given in Theorem~\ref{thm:f-y}, can be rewritten in a vector form. Doing so yields an equivalent characterization of the existence of an EFX allocation as the fixed point of a continuous map. This is formalized in the following theorem. 
\begin{theorem}\label{thm:T-map}
Given $y\in \mathbb{R}^{mn}$ and agents' valuations $\{v_{\ell i}\}$, let \footnote{$e_{ij} = e_j - e_i$, where $e_i$ and $e_j$ denote the $i$-th and $j$-th standard basis vectors in $\mathbb{R}^n$, respectively.}  
\begin{align}\label{eq:A-h}
h(y_{\ell}) &= \max_{r} y_{\ell r} \ \forall \ell,\cr
A_{kj}(y) &= \max_{i: i \ne j} \sum_{\ell \ne k} \big( h(y_{\ell} + v_{\ell i} e_{ij}) - h(y_{\ell}) \big)\ \forall k,j.
\end{align}
Consider the vector map $T : \mathbb{R}^{mn} \to \mathbb{R}^{mn}$ whose $(k,j)$-th coordinate is defined by
\begin{align}\label{eq:T-coordinate-def}
T_{kj}(y) = \min\Big\{ y_{kj},\ h(y_k) - A_{kj}(y) \Big\},
\end{align}
Then, an EFX allocation exists if and only if $T$ has a fixed point.\footnote{$y^* \in \mathbb{R}^{mn}$ is a fixed point for $T$ if $T(y^*) = y^*$.}
\end{theorem}
\begin{proof}
From Theorem \ref{thm:f-y}, we know that an EFX allocation exists if and only if  $\inf_{y\in \mathbb{R}^{mn}}f(y)\leq 0$, where
\begin{align*}
f(y)=\max_{i\neq j, k}\Big(y_{kj}+\sum_{\ell\neq k}\max\Big\{y_{\ell i}-v_{\ell i}, y_{\ell j}+v_{\ell i}, \max_{r\neq i,j }y_{\ell r}\Big\}\Big)-\sum_{\ell}\max_{r} y_{\ell r}. 
\end{align*}
Therefore, from \eqref{eq:A-h}, we have
\begin{align*}
h(y_{\ell}+v_{\ell i}e_{ij})=\max\Big\{y_{\ell i}-v_{\ell i}, y_{\ell j}+v_{\ell i}, \max_{r\neq i,j }y_{\ell r}\Big\}.
\end{align*}
Now, we can rewrite $f(y)$ in a slightly different form as
\begin{align}\nonumber 
f(y)&=\max_{i\neq j, k}\Big(y_{kj}+\sum_{\ell\neq k}h(y_{\ell}+v_{\ell i}e_{ij})\Big)-\sum_{\ell}h(y_{\ell})\cr
&=\max_{i\neq j, k}\Big(y_{kj}-h(y_k)+\sum_{\ell\neq k}\big(h(y_{\ell}+v_{\ell i}e_{ij})-h(y_{\ell})\big)\Big)\cr
&=\max_{j, k}\Big(y_{kj}-h(y_k)+\max_{i: i\neq j}\sum_{\ell\neq k}\big(h(y_{\ell}+v_{\ell i}e_{ij})-h(y_{\ell})\big)\Big)\cr
&=\max_{j, k}\Big(y_{kj}-h(y_k)+A_{kj}(y)\Big),
\end{align}
where the last equality follows by the definition of $A_{kj}(y)$. Therefore, an EFX allocation exists if and only if there exists $y \in \mathbb{R}^{mn}$ such that $y_{kj} - h(y_k) + A_{kj}(y) \leq 0\ \forall j,k$, or equivalently, if there exists $y \in \mathbb{R}^{mn}$ such that $
\min\{y_{kj}, h(y_k) - A_{kj}(y)\} = y_{kj}\ \forall j,k$. As a result, if we define the vector map $T:\mathbb{R}^{mn} \to \mathbb{R}^{mn}$ as in~\eqref{eq:T-coordinate-def}, then an EFX allocation exists if and only if $T$ has a fixed point in $\mathbb{R}^{mn}$.
\end{proof}
  
Currently, we do not know whether \( T \) admits a fixed point. However, in all numerical experiments we conducted, such a fixed point consistently appeared to exist. This observation motivated us to investigate whether the existence of a fixed point for \( T \) could be derived from known fixed-point theorems, such as the Brouwer's fixed-point theorem. To that end, we first list some properties of the mapping \( T \) that will be useful in studying the existence of its fixed points.

\begin{itemize}
    \item $T$ is a continuous map. This follows directly from the fact that each coordinate of $T$ is expressed using the (continuous) $\min$ and $\max$ functions.
    
    \item $|A_{kj}(y)| \leq V$ for all $y \in \mathbb{R}^{mn}$ and all $j, k$, where $V = \sum_{\ell,i} v_{\ell i}$. Indeed, for any $y \in \mathbb{R}^{mn}$,
    \begin{align}\label{eq:A-bound-V}
    |A_{kj}(y)| \leq \max_{i \ne j} \sum_{\ell \ne k} \left| h(y_{\ell} + v_{\ell i} e_{ij}) - h(y_{\ell}) \right|
    \leq \max_{i \ne j} \sum_{\ell \ne k} v_{\ell i}
    \leq V.
    \end{align}
    
    \item As shown in the proof of Theorem~\ref{thm:f-y}, if an EFX allocation exists, then $f(y) \leq 0$ for some $y \in [-M, 0]^{mn}$, where $M$ is a sufficiently large constant (e.g., $M = 2V + 1$).
\end{itemize}

According to the above properties, without loss of generality, we may restrict the domain of \( T \) to the compact and convex box \( [-M, 0]^{mn} \), for a sufficiently large finite \( M \), and search for the existence of a fixed point within that box. Therefore, we can say that the continuous mapping $T : [-M, 0]^{mn} \to [-M - V, 0]^{mn}$ has a fixed point if and only if an EFX allocation exists. In particular, by replacing \( T \) with its normalized version \( \frac{1}{M}T(My) \),\footnote{The normalized map can be viewed as the same map \( T \), except that agents' valuations \( v_{\ell i} \) are replaced by their normalized values \( \frac{v_{\ell i}}{M} \).} and noting that \( y^* \) is a fixed point of \( T \) if and only if \( y^* \) is a fixed point of its normalized version, we may assume that the domain of \( T \) is \( [-1, 0]^{mn} \), where the agents' valuations are replaced by their normalized values \( \frac{v_{\ell i}}{M} \). This brings us to the following conjecture:

\smallskip 
\begin{conjecture}\label{conj-map}
Let \( \epsilon > 0 \) be an arbitrarily small number, and without loss of generality, assume that \( V = \sum_{\ell,i} v_{\ell i} \leq \epsilon \).\footnote{Otherwise, we can replace each \( v_{\ell i} \) with \( \frac{v_{\ell i}}{M} \), where \( \frac{V}{M} \leq \epsilon \).} Let $T: [-1, 0]^{mn} \to [-1 + \epsilon, 0]^{mn}$ be a continuous map defined by
\begin{align*}
T_{kj}(y) = \min\left\{ y_{kj},\ h(y_k) - A_{kj}(y) \right\} \quad \forall\, k, j,
\end{align*}
where $h(y_{\ell})$ and $A_{kj}(y)$ are defined as in \eqref{eq:A-h}. Then \( T\) admits a fixed point.  
\end{conjecture}

\smallskip
It is worth noting that the mapping \( T \) in the above conjecture satisfies all the conditions of the Brouwer’s fixed-point theorem, except for the self-mapping property, which is violated by an arbitrarily small amount \( \epsilon > 0 \). This suggests that, perhaps through an appropriate nonlinear transformation, one could slightly bend the mapping \( T \) near the boundary of its domain to satisfy the conditions of Brouwer’s fixed-point theorem—where the fixed point of the perturbed map would also serve as a fixed point for the original map \( T \). This is precisely the approach we will pursue in the following section.

\subsection{An Approximate Map with a Guaranteed Fixed Point}

Building on the framework developed above, we begin this section by re-expressing the EFX conditions in a slightly modified form, framed as a fixed point of an alternative discrete vector map. Although this map is discontinuous on the boundaries of its domain, we approximate it closely using a perturbed continuous map that satisfies Brouwer's fixed point theorem, and therefore always admits a fixed point. While we are currently unable to prove that every fixed point of the perturbed map is also a fixed point of the original discrete map, we provide strong evidence suggesting this is the case. To describe the structure of the approximate map, we fist consider the following lemma. 

\smallskip
\begin{lemma}\label{lemm:discrete-map}
Consider a vector map $T':[-M, 0]^{mn}\to [-M, 0]^{mn}$ whose $(k,j)$th coordinate is given by $T'_{kj}(y)=\min\{y_{kj}-h(y_k), -A_{kj}(y)\cdot\boldsymbol{1}_{\{h(y_k)=0\}}\}$, where $\boldsymbol{1}_{\{\cdot\}}$ is the indicator function, and $h(y_k), A_{kj}(y)$ are defines as in \eqref{thm:T-map}. Then, an EFX allocation exists if and only if $T'$ has a fixed point.  
\end{lemma}
\begin{proof}
Let us recall from the proof of Theorem~\ref{thm:T-map} that an EFX allocation exists if and only if there exists a vector $y \in \mathbb{R}^{mn}$ such that
\begin{align}\label{eq:constraints-A}
    y_{kj} - h(y_k) + A_{kj}(y) \leq 0 \quad \forall\, k, j.
\end{align}

Assume there exists a solution $y \in [-M, 0]^{mn}$ satisfying \eqref{eq:constraints-A}, where,  without loss of generality, we may take $h(y_k) = 0$ for all $k$. If this is not the case, we can increase the maximum entry in each row of $y$ to zero, which can only decrease the expressions $y_{kj} - h(y_k) + A_{kj}(y)$, thereby preserving the feasibility of the constraints in \eqref{eq:constraints-A}. We refer to such a feasible solution as a \emph{maximal} feasible solution. Let $y$ be a maximal feasible solution to \eqref{eq:constraints-A}. Then we have
\begin{align*}
    \min\left\{y_{kj} - h(y_k),\ -A_{kj}(y) \cdot \boldsymbol{1}_{\{h(y_k) = 0\}}\right\}
    &= \min\left\{y_{kj} - h(y_k),\ -A_{kj}(y)\right\} \\
    &= y_{kj} - h(y_k) = y_{kj} \quad \forall\, k, j,
\end{align*}
which shows that $y$ must be a fixed point of $T'$.

Conversely, suppose that $y$ is a fixed point of $T'$. Then, $$\min\{y_{kj} - h(y_k),\ -A_{kj}(y) \cdot \boldsymbol{1}_{\{h(y_k) = 0\}}\} = y_{kj}\ \forall k,j.$$ If the minimum is attained strictly at the first argument, then $y_{kj} - h(y_k) = y_{kj}$, which implies $h(y_k) = 0$. Since $y_{kj} - h(y_k) \leq -A_{kj}(y) \cdot \boldsymbol{1}_{\{h(y_k) = 0\}} = -A_{kj}(y)$, it follows that the constraint in~\eqref{eq:constraints-A} is satisfied. On the other hand, the minimum cannot be attained strictly at the second argument. To see this, suppose $-A_{kj}(y) \cdot \boldsymbol{1}_{\{h(y_k) = 0\}} = y_{kj}$. If $h(y_k) = 0$, then $-A_{kj}(y) = y_{kj}$, so $y_{kj} - h(y_k) = y_{kj} = -A_{kj}(y)$, contradicting the assumption that the minimum is strictly at the second argument. Alternatively, if $h(y_k) \neq 0$, then the indicator term is zero, and we must have $-A_{kj}(y) \cdot \boldsymbol{1}_{\{h(y_k) = 0\}} = 0 = y_{kj}$. However, this implies that $h(y_{k}) = 0$, which is a contradiction. Therefore, the minimum must always be attained at the first argument, which implies $h(y_k) = 0$ and thus the constraints in~\eqref{eq:constraints-A} are satisfied. Hence, any fixed point of $T'$ corresponds to a feasible solution to~\eqref{eq:constraints-A}.
\end{proof}

Although the map $T'$ satisfies the self-containedness condition of Brouwer's fixed point theorem, it unfortunately violates the continuity condition due to the presence of the indicator function. This motivates a slight modification of the map $T'$, in which its image is slightly curved near the boundary of its domain to ensure the continuity condition is satisfied. This leads to a new map $\tilde{T}$, for which the existence of a fixed point is guaranteed. More specifically, we consider a continuous approximation of $T'$ where the indicator function $\boldsymbol{1}_{\{h(y_k)=0\}}$ is replaced by the exponential function $e^{h(y_k)}$. Therefore, we define the perturbed vector map $\tilde{T} : [-M, 0]^{mn} \to [-M, 0]^{mn}$ as
\begin{align}\nonumber
\tilde{T}_{kj}(y) = \min\left\{ y_{kj} - h(y_k),\ -A_{kj}(y) e^{h(y_k)} \right\} \quad \forall k, j.
\end{align}

Clearly, $\tilde{T}$ is a continuous map that maps the compact and convex box $[-M, 0]^{mn}$ into itself. Thus, by Brouwer's fixed point theorem, it admits a fixed point. In fact, using the same argument as in the proof of Lemma~\ref{lemm:discrete-map}, if the system~\eqref{eq:constraints-A} has a feasible solution (i.e., if an EFX allocation exists), then the maximal feasible solution must be a fixed point of $\tilde{T}$.

We aim to argue that a fixed point of $\tilde{T}$ serves as a feasible solution (or at least a good approximate feasible solution) to~\eqref{eq:constraints-A}. To that end, let $y = \tilde{T}(y)$ be a fixed point of $\tilde{T}$. Then, either the minimum in the definition of $\tilde{T}$ is achieved by its first argument, in which case we have $$y_{kj} - h(y_k) = y_{kj} \leq -A_{kj}(y) e^{h(y_k)}.$$ This implies that $h(y_k) = 0$, and thus $y_{kj} - h(y_k) \leq -A_{kj}(y)$, satisfying the constraints~\eqref{eq:constraints-A}. Otherwise, if the minimum is achieved by its second argument, then we have $y_{kj} - h(y_k) > -A_{kj}(y) e^{h(y_k)} = y_{kj}$. Therefore, $h(y_k) < 0$, $y_{kj} < 0$, and it must hold that $e^{h(y_k)} = -\frac{A_{kj}(y)}{y_{kj}}$. However, note that this last equality must hold for every $j$ in the $k$th row. Otherwise, if for some $(k, j') \neq (k, j)$ the first case occurs, then we would have $h(y_k) = 0$, contradicting the fact that $h(y_k) < 0$. As a result, in the case where the minimum is achieved by the second argument for some $k$, we have 
\begin{align}\label{eq:final} h(y_k)<0,\ y_{kj}<0, \ e^{-h(y_k)}=-\frac{A_{kj}(y)}{y_{kj}}\ \forall j. 
\end{align}
These impose a set of stringent nonlinear equality constraints among the elements of row $k$, which we believe cannot hold simultaneously. For instance, consider the index $j$ for which $y_{kj} = h(y_k)$ in the relation {eq:final}. Then, it is easy to see that $h(y_k) = O(\ln V)$ and therefore $h(y_k)$ will be close to zero—an ideal case.

In fact, in all the numerical experiments we conducted, we observed that every fixed point of the proxy map $\tilde{T}$ satisfied all the constraints \eqref{eq:constraints-A}, thereby implying the existence of an EFX allocation. We leave it as a future direction to formally prove that the constraints in~\eqref{eq:final} do not admit a feasible solution, or, if they do, that such a solution can be closely rounded to a feasible solution of~\eqref{eq:constraints-A}.

\section{Conclusions}\label{sec:con} 

\noindent
In this work, we presented a novel approach to the EFX allocation existence problem—one of the most important open questions in the fair division of indivisible goods—by bridging it with continuous fixed-point theory and difference-of-convex (DC) optimization. By extending the discrete EFX constraints into a continuous domain via randomized rounding, we established that the existence of an EFX allocation is equivalent to the optimal value of a corresponding DC program being nonpositive. We further showed that this program reduces to the minimization of a piecewise linear concave function over a polytope, and that its optimality is linked to the existence of a fixed point of a carefully defined continuous vector map. Although the original map nearly satisfies the conditions of Brouwer’s fixed-point theorem, we addressed the minor violation of self-containedness by introducing a perturbed version that always admits a fixed point. This fixed point serves as a proxy for an actual EFX allocation. Overall, our framework not only sheds new light on the theoretical existence of EFX allocations, but also offers practical pathways for their computation through tools from nonlinear optimization and fixed-point theory.

There are several promising directions for future work. One is to investigate whether a better randomized rounding scheme exists that yields a tighter description of the corresponding vector map, potentially enabling one to show that it satisfies a known fixed-point theorem. Another direction is to formally prove Conjecture~\ref{conj-map}. A third avenue is to establish that any fixed point of the approximate map $\tilde{T}$ either corresponds directly to an EFX allocation or can be efficiently rounded to one. Finally, recent advances in duality theory for DC programming could be leveraged to either prove the nonpositivity of the optimal objective value of the proposed DC program or to develop more computationally efficient algorithms for computing an EFX allocation, when one exists.

\bibliographystyle{IEEEtran}
\bibliography{thesisrefs}

%\appendices

\end{document}